\documentclass[11pt]{article}
\usepackage[margin=1in]{geometry}    
\usepackage{gensymb}
\usepackage{graphicx}
\usepackage{amssymb,amsmath}
\usepackage{amstext,amsthm, amssymb,amsfonts}
\usepackage{epstopdf}
\usepackage{placeins}
\usepackage{multirow}
\usepackage[normalem]{ulem}
\useunder{\uline}{\ul}{}

\newtheorem{theorem}{Theorem}[section]

\newtheorem{corollary}[theorem]{Corollary}
\newtheorem{definition}[theorem]{Definition}

\providecommand{\keywords}[1]
{
  \small	
  \textbf{\textit{Keywords---}} #1
}

\makeatletter
\providecommand*{\Dashv}{%
  \mathrel{%
    \mathpalette\@Dashv\vdash
  }%
}
\newcommand*{\@Dashv}[2]{%
  \reflectbox{$\m@th#1#2$}%
}

\title{Implication Avoiding Dynamics for Externally Observed Networks}

\date{}        
\begin{document}
\maketitle

\centerline{\scshape Joel Nishimura and Oscar Goodloe}
\medskip
{\footnotesize
 \centerline{School of Mathematical and Natural Sciences}
   \centerline{ Arizona State University, Glendale, AZ 85306-4908, USA}
}

\begin{abstract}

Previous network models have imagined that connections change to promote structural balance, or to reflect hierarchies.  We propose a model where agents adjust their connections to appear credible to an external observer.  In particular, we envision a signed, directed network where positive edges represent endorsements or trust and negative edges represent accusations or doubt, and consider both the strategies an external observer might use to identify credible nodes and the strategies nodes might use to then appear credible by changing their outgoing edges.  First, we establish that an external observer may be able to exactly identify a set of `honest' nodes from an adversarial set of `cheater' nodes regardless of the `cheater' nodes' connections.  However, while these results show that an external observer's task is not hopeless, some of these theorems involve network structures that are NP-hard to find. Instead, we suggest a simple heuristic that an external observer might use to identify which nodes are not credible based upon their involvement with particular implicating edge motifs. Building on these notions, and analogously to some models of structural balance, we develop a discrete time dynamical system where nodes engage in implication avoiding dynamics, where inconsistent arrangements of edges that cause a node to look 'suspicious' exert pressure for that node to change edges. We demonstrate that these dynamics provide a new way to understand group fracture when nodes are worried about appearing consistent to an external observer. 

\end{abstract}

\keywords{
trust on networks, network dynamics, structural balance, signed directed network, social networks} 

\section{Introduction}
\label{sec:intro}
Social networks often capture a large and varied range of interactions or relationships and are similarly created by a variety of different forces. Two well-known such forces are: structural/social balance \cite{cartwright1956structural,heider1946attitudes}, where a friend of a friend and an enemy of an enemy each become friends while a friend of an enemy and an enemy of a friend become enemies; and status \cite{leskovec2010signed,leskovec2010predicting}, where one esteems those of higher rank and disparages those of lower rank. Both structural balance and status describe forces that shape the construction of signed and signed directed graphs respectively. We propose another possible force, credibility, where contradictory judgments lead a node to appear suspicious, and we use this force to suggest an `implication avoiding' network dynamics.   

An investigation of credibility begs the question: credibility according to whom? In order to develop a theory of credibility we imagine the following setting: 
individuals endorse and accuse each other, and from this, an external observer will make decisions regarding which nodes are credible.
For instance, imagine a portion of a the class has cheated on an exam, and the instructor, interested in discovering the cheaters, privately asks students to endorse and/or accuse the integrity of their classmates. Alternatively, one could imagine any of the many settings where an investigator or new-comer arrives at a network and must determine credibility based on a web of endorsements and accusations. How could such an external observer determine credibility, and how might individuals behave in order to manipulate that determination?

To address these questions we consider the idealized scenario of a directed signed network made of endorsements (positive edges) and accusations (negative edges), where nodes are either `honest' or are adversarial `cheaters' and the external observer's goal is to correctly identify the honest nodes based solely on the network structure. We further assume that honest nodes will correctly endorse other honest nodes and only accuse cheaters, but the cheaters are free to endorse or accuse any node in an attempt to confuse the external observer. 
This idealized setting allows us to clearly formulate conditions when an observer can exactly identify the honest nodes regardless of the strategies of the cheaters. 

While this work establishes that the observer's task is not wholly unreasonable, it is unlikely that a person would casually implement the strategies implicit in our theorems. In contrast, we note that there are five small inconsistent motifs that immediately implicate a node as belonging to $C$. Since checking small motifs requires neither great insight nor computation, this simple heuristic suggests an `implication avoiding dynamics' (IAD) where nodes attempt to avoid being implicated by inconsistent motifs. 

While we have motivated IAD by considering the idealized setting where `honest' nodes make no errors, the internal logic of IAD should remain pressing whenever nodes are concerned about an external observer checking for inconsistent motifs, regardless of whether `honest' nodes really are error-free or if there is even a ground truth preferred group at all. More generally, the assumptions of IAD may well apply to situations where an observer expects consistency from any sort of `in-group,' and the nodes wish to be labeled as in that group.

Implication avoiding dynamics predicts sparse endorsements among an initially inconsistent out-group, and community fracture, where even a single internal accusation has a chance to cascade into a web of accusations; both are behaviors inexplicable to structural balance and status models. 

Indeed, IAD provides an interesting complement to existing theories of structural balance dynamics. While structural balance, which operates on undirected graphs, is fundamentally a process which involves dyadic coordination (positive edges imply mutual approval), IAD involves only unilateral changes, where a node changes an outgoing directed edge. From this perspective, IAD describes a faster and externally instigated set of dynamics, while structural balance describes a slower internally motivated set of dynamics. These differences manifest themselves in the possible stationary/balanced states.
For instance, in structural balance the only balanced complete graphs are those with exactly two\footnote{One faction may be empty.} opposing factions, while IAD can support any number mutually opposing factions, some of which would be considered jammed states in structural balance \cite{marvel2009energy,antal2005dynamics}. Similarly, IAD is able to sustain more structures in sparse graphs than structural balance, either in it's triangle formulation or its more general cycle definition \cite{wasserman1994social}. One reason for these differences is that in traditional structural balance a triangle with three negative edges is unbalanced, while it is not in IAD, consistent with recent empirical work where such triangles are rare \cite{rawlings2017structural,isakov2019structure} and work which consider whether structural balance should be split up into a collection of different sub-dynamics \cite{rawlings2017structural}.   
Other work in structural balance has formulated an elegant continuous extension: $\frac{dX}{dt}=X^2$ \cite{marvel2011continuous}, more strategic notions of single time step dynamics \cite{van2011micro} and connections with monotonic dynamical systems \cite{10.1371/journal.pone.0038135}, all of which could be fruitful directions for further study of IAD.

Another complementary set of dynamics is that of status theory \cite{leskovec2010predicting,leskovec2010signed}. In comparison to IAD, status theory imagines a hierarchical node embedding where positive edges are directed up the hierarchy and negative edges are directed down it. In status theory, a reciprocated edge should have one negative direction and on positive direction (according to which direction has increasing status), while IAD suggests this is the least stable motif and instead predicts that edges should be reciprocated in kind.     

Our work also has connections with trust-related algorithms. In many online settings, determining the trustworthiness of nodes in a network is a practical problem that platform designers must address, and this has led to a number of algorithmic approaches. In contrast to our setting, the online setting typically cannot assume that there is a core or consistent set of honest nodes, yet they are similarly worried about the effects of strategic adversaries. There is some evidence that the approaches which are good at aggregating information are typically susceptible to manipulation and vice versa \cite{tang2010hybrid}. One kind of attack that is particularly important in online settings is the creation of new accounts created solely to bolster the reputation or trust of another user.  These so-called `sybil' accounts can present a distinct challenge to otherwise classical approaches, and have necessitated a series of developments specific for this threat  \cite{cheng2005sybilproof,resnick2009sybilproof,liu2016personalized}. While we explicitly frame our question in terms of adversarial nodes, we do not consider manipulations which create new nodes, such as creating sybils, and thus our results are less relevant for online settings where there are negligible transaction costs. 

Other recent approaches consider iterative notions of trust, such as the `bias-deserve' model \cite{mishra2011finding}, the `fairness-goodness' model \cite{kumar2016edge} and its extension \cite{kumar2018rev2}. An interesting aspect of these approaches is that they split notions of trust into two approximately orthogonal notions, roughly corresponding to whether one can trust a node's outgoing edges and how one should interpret a node's incoming edges. While these methods perform very well at predicting edges in cross validation tasks, they require additional node level information in order to distinguish strategic malicious users from ordinary ones, as explicitly done in \cite{kumar2018rev2}. 

Ideally, these works on network dynamics and trust will expand the vocabulary with which the various phases of growth and evolution of networks can be described. Towards that end, IAD plays an important role as a description and explanation of community fracture, especially when that fracture is instigated by an external event, such as external observation.

\section{A Network Implication Game}
\label{sec:implication_game}
We consider a setting where an `observer' $O$, attempts to correctly classify the members of $V = \{H, C\}$ composed of `honest' players $H$, and `cheaters' $C$, where members of both $H$ and $C$ aim to be classified as members of $H$. In this game, each member $u$ of $V$ produces some number of directed edges of the form $u\to v$, which represent endorsements, and directed edges of the form $u\Dashv v$, which represent accusations. Members of $H$ are constrained to only endorse other members of $H$ and accuse members of $C$, while members of $C$ have no such constraints. Together, the members of $C$ and $H$ create a signed network $G=\{V,E,A\}$, with nodes $V$, positive endorsements $E$ and negative accusations $A$. Using only the structure of $G$ the observer attempts to infer which players belong to $C$. 

In this game, the most pressing question is whether it is ever possible for the administrator to identify $H$, or if instead, strategic behavior of $C$ can always confuse the administrator. Indeed, if $|C|\ge|H|$ then a subset of $C$ can simply mirror the structure of $H$, rendering $H$ indistinguishable from their doppelgangers in $C$. Thus if $|C|\ge|H|$, then $C$ will not be distinguishable by an observer, at least one whose only information is the network structure. 
We thus restrict our attention to the scenario where $|H|>|C|$. The following definitions clarify the distinguishing features of $H$.

\begin{definition}
A set $Q$ is self-consistent if there does not exist $u,v \in Q$ such that $u \Dashv v$.
\end{definition}

\begin{definition}
A set $Q$ is insular if there does not exist any $u \in Q$, $v \notin Q$ such that $u \to v$.
\end{definition}

Notice, that $H$ is always self-consistent and insular by assumption, but it may not be the only or the largest such set. In the following theorems we explore the sufficient conditions on the edges of $H$ such that $H$ is the largest self-consistent and/or insular set.  

To begin, let $G_H$ and $G_C$ be the implied subgraphs of $G$, composed of nodes in $H$ and $C$ respectively, and let $G^\to$ and $G^\dashv$ be the subgraphs of $G$ with only endorsements (deleting accusations) and accusations (deleting endorsements) respectively. Additionally, for some set $U$, let $\rho(U)$ be the set of endorsement-upstream nodes such that: $v \in \rho(U)$ if and only if there exists a path of endorsements $v\to\ldots \to u \in U$. Let $\rho_H(U)$ denote the nodes endorsement-upstream of $U$ in $G_H$, implying that $\rho_H(U)\subseteq H$ if $U\subseteq H$. Similarly, let $\sigma(U)$ denote the endorsement downstream nodes such that: $v\in \sigma(U)$ if and only if there exists a path of endorsements $u \to\ldots \to v$ for some $u\in U$. 

Notice that with this notation that $Q$ is insular if and only if $\sigma(Q)\subseteq Q$.  Similarly, if $Q$ is insular, $u\in Q$ and $u\in \rho(v)$, then $v\in Q$. This assertion also yields the contrapositive statement: if $Q$ is insular and $u\not \in Q$ then $\rho(u)\cap Q = \varnothing$. Thinking in these terms suggests the following theorem.

\begin{theorem}
If $|H|>|C|$ and $G_H$ contains a strongly connected component $U$ such that $|U|>|C|$ and $H\subseteq \sigma(U)$, then the observer can exactly determine $H$ and $C$. \label{thm:errorFreeStrong}
\end{theorem}
\begin{proof} Since $H$ is insular, a strongly connected component is either entirely in $H$ or in $C$. As $|U|>|C|$, then $U\subseteq H$ as is $\sigma(U)$ which, by assumption, is all of $H$. Thus, regardless of the edges of $C$, the observer can exactly determine $H$ by identifying the largest strongly connected component and classifying exactly $\sigma(U)$ as members of $H$.  \end{proof}

Clearly one way the nodes in $H$ can satisfy Theorem \ref{thm:errorFreeStrong} is by producing enough edges to percolate themselves into a strongly connected subgraph. In other words, this situation epitomizes the scenario where members of $H$ avoid being falsely accused by creating a network of secure endorsements. Notice that for large graphs, $H$ can achieve a large strongly connected component with high probability simply by creating random edges so long as each node has an expected degree of at least $(1+\epsilon)\log(|H|)$ for $\epsilon>0$. 
While Theorem \ref{thm:errorFreeStrong} made $H$ determinable using only endorsements, a similar result can be achieved using only accusations, for which we need the following definition. 

\begin{definition}
For sets $S$ and $T$, $T \Dashv S$ if every node in $T$ accuses some node in $S$ and every node in $S$ is accused by at least one node in $T$.
\end{definition}

\begin{corollary}
If for every subset $S\subseteq C$ there exists $U\subseteq H$ such that $U\Dashv S$ and $|U|>|S|$ then the largest self-consistent group is $H$.
\label{cor:outnumber}
\end{corollary}

This intuitive result is in fact a special case of Theorem \ref{thm:outnumberingProperty}, which will generalize the statement to graphs with node weights. While it might seem that the hypothesizes of Corollary \ref{cor:outnumber} both require a very large number of accusations and would be difficult to verify there are some simple constructions in $H$ that fulfill the hypothesizes using only $|C|+1$ accusations.  

\begin{corollary}
If there exists a path of the form $h_1 \Dashv c_1 \vdash h_2 \Dashv c_2 \ldots \vdash h_k $, where $c_i\in C$, $h_i\in H$, the $c_i$ and $h_i$ are unique and $k = |C|+1$, then the largest self-consistent group is $H$.
\label{thm:HamiltonianPath}
\end{corollary}
\begin{proof} Notice that any collection of cheaters can be expressed as the union of sets of consecutively numbered cheaters between some $c_a$ and $c_b$. Thus there exists $h_{a}\Dashv c_a, h_i \Dashv c_i, h_{b}\Dashv c_b$ and $h_{b+1}\Dashv c_b$, for $a<i\le b$ and therefore for any collection of cheaters we can immediately find a greater collection of honest players accusing them. Thus by Corollary \ref{cor:outnumber}, $H$ is the largest self-consistent group. \end{proof} 

Notice that the hypotheses of Corollary \ref{thm:HamiltonianPath} essentially amount to the existence of a Hamiltonian bipartite undirected subgraph of $G^{\Dashv}$, restricted to the accusations from $H$ into $C$. Consequently, if the honest players can spread their accusations broadly enough across $C$ they can ensure that the largest self-consistent group is exactly $H$. We now show the more general Theorem \ref{thm:outnumberingProperty} for graphs with node weights, where each node  $u\in V$ has weight $w_u\ge 0$.

\begin{theorem}
If for every $S \subset C$ there exists $U \subset H$ such that $U \Dashv S$ and $W_U = \sum_{u\in U} w_u$ is greater than $W_S = \sum_{v \in S} w_v$, then the most heavily weighted self-consistent group of nodes is $H$. 

\label{thm:outnumberingProperty}
\end{theorem}
\begin{proof}Let $Q$ be the most heavily weighted self-consistent group of nodes. Assume to the contrary that $Q\cap C$ is nonempty and let $S = Q \cap C$. Since $S \subset C$, then by our assumption there exists $U \subset H$ such that $U \Dashv S$ and $W_U>W_S$. Since $Q$ is consistent, then $U\cap Q = \varnothing$ and thus $\hat{Q} = U \cup (Q\setminus S)\in H$. Additionally, $\hat{Q}$ is self-consistent and has weight $W_Q+W_U-W_S>W_Q$, contradicting the assumption that $Q$ is the largest self-consistent set of nodes. Thus $Q\subseteq H$, and since $W_H\ge W_Q$, then $Q=H$. 
\end{proof}

The benefit of proving Theorem \ref{thm:outnumberingProperty} for graphs with node weights is that it can be applied to the condensation of $G^\to$, where each node represents a strongly connected component of $G^\to$ and its weight is the component's size.  For instance, if $H$ is split between two strongly connected components, $A$ and $B$, then it may not be identifiable from Theorem \ref{thm:errorFreeStrong}, but if a member of $A$ and a member of $B$ are each able to accuse all of $C$, then $H$ can be determined according to Theorem \ref{thm:outnumberingProperty}. 

We can further generalize the approach of Theorem \ref{thm:outnumberingProperty} by considering more fully how endorsements and accusations can interact, and how endorsements internal to $C$ can augment the accusations of $H$.

\begin{theorem}
 Let $Q$ be the largest self-consistent, insular set in our network. If for any set of cheaters $S$ there exists a set $K$ such that
\begin{enumerate}
    \item $S \subset \rho(K)\cup K$,
    \item $S \cap \rho(k) \neq \varnothing \quad \forall k \in K$, and
    \item there exists $U \subset H$ such that $U \Dashv K$ and $|\rho_H(U)| > |\rho(K)|$,
\end{enumerate}
then $Q = H$.
\label{thm:tree}
\end{theorem}

\begin{proof} Since $Q$ is the largest self-consistent, insular set, if $Q \neq H$, then $Q_C = Q\cap C \neq \varnothing$. By assumption, there exists $K$ such that $Q_C \subseteq \rho(K)\cup K$. Since $Q_C \cap \rho(k) \neq \varnothing$ for all $k \in K$, then $K \subseteq Q$. By assumption there also exists $U \subset H$ such that $U \Dashv K$. Self consistency of $Q$ implies that $U \cap Q = \varnothing $, and insularity of $Q$ then implies that $\rho_H(U)\cap Q = \varnothing$. Thus, $\rho_H(U) \cup (Q\setminus Q_C) \subseteq H$, and $|Q|< |\rho_H(U) \cup (Q\setminus Q_C)| \le |H|$, leading to the contradiction that $H$ is a strictly larger self-consistent insular set than $Q$.   \end{proof}

The approach taken in the proof of Theorem~\ref{thm:tree} is to ensure that the inclusion of cheaters is always non-optimal in the construction of the largest possible insular set. That is, $Q$ cannot be the largest insular set if it includes any subset of $\rho(V)$ for $V \subset C$. Conceptually, one way to think of the result in Theorem~\ref{thm:tree} is that strength of the accusations of $U$ is weighted by the number of nodes with endorsements leading to $U$. For example, a strongly connected group where each node makes a few accusations is, at least according to the hypothesizes of Theorem~\ref{thm:tree}, functionally the same as if that group had a single member who made the same accusations.  

However, while Theorem~\ref{thm:tree} and Theorem~\ref{thm:outnumberingProperty} give a clear signal for an external observer to search for, finding that signal can be NP-hard.  Namely, Theorem~\ref{thm:outnumberingProperty} implies that the honest nodes may be identified by finding the largest self-consistent subset of nodes, but this is clearly NP-hard as it is equivalent to finding the largest size clique in the complement of $G^{\Dashv}$. Similarly, in a graph without endorsements, the largest self-consistent insular set is exactly the largest self-consistent set, and thus finding the object that Theorem~\ref{thm:tree} suggests is also NP-hard. Thus, while these theorems prove that there are plenty of situations where the graph of accusations and endorsements allow for perfect identification of the honest nodes regardless of the behavior of the cheaters, they lack clear direction as to how an external observer with limited resources would actually investigate a medium or large sized network.

\section{Inconsistent Motifs}
While Theorem~\ref{thm:tree} shows that under the right circumstances it is possible for an external observer to find the honest nodes, it does not suggest an efficient algorithm. A natural next step is to consider heuristics that an observer could use to evaluate the credibility of specific agents. Since a node in $H$ neither accuses other nodes in $H$ nor endorses nodes in $C$, there are several edge motifs honest nodes will not be involved in. As mentioned in the previous section, if $u\in H$ then $\sigma(u)\subseteq H$ and thus $u$ will neither accuse nor be accused by any node $v\in \sigma(u)$.  Similarly, for $u\in H$ no two nodes $v,w\in\sigma(u)$ will accuse each other.
 
Of these inconsistent motifs, the more important and most easily detected are those on edges and triangles. For example, if $u\to v$ and $v \Dashv u$, then $u\not \in H$.  Similarly, on three nodes there are three inconsistent edge motifs, as shown in Figure \ref{fig:susp_motifs}. For Type I arrangements, consistency implies that the highlighted node should agree with the judgment of its endorsed node, yet it directly conflicts with that judgment. Similarly, in Type II arrangements, the highlighted node should agree with the accusation of its endorsed nodes, but it does not. Finally, in a Type 3 arrangement the highlighted node causes itself to look suspicious by lending credibility to a group of nodes which ultimately accuse it. 

Since any node implicated by an inconsistent motif is not in $H$, an observer can conduct a preliminary screening of the network, using inconsistent motifs to classify any implicated node $u$ and it's upstream endorsers, $\rho(u)$ to $C$. As described later in Sections \ref{sec:local_iad} and \ref{sec:strong_iad}, this preliminary pass necessarily results in a network without inconsistent motifs and such networks have structures that simplify, without necessarily speeding up, finding the largest self-consistent and insular sets.

\begin{figure}
    \centering
    \includegraphics[width=\textwidth]{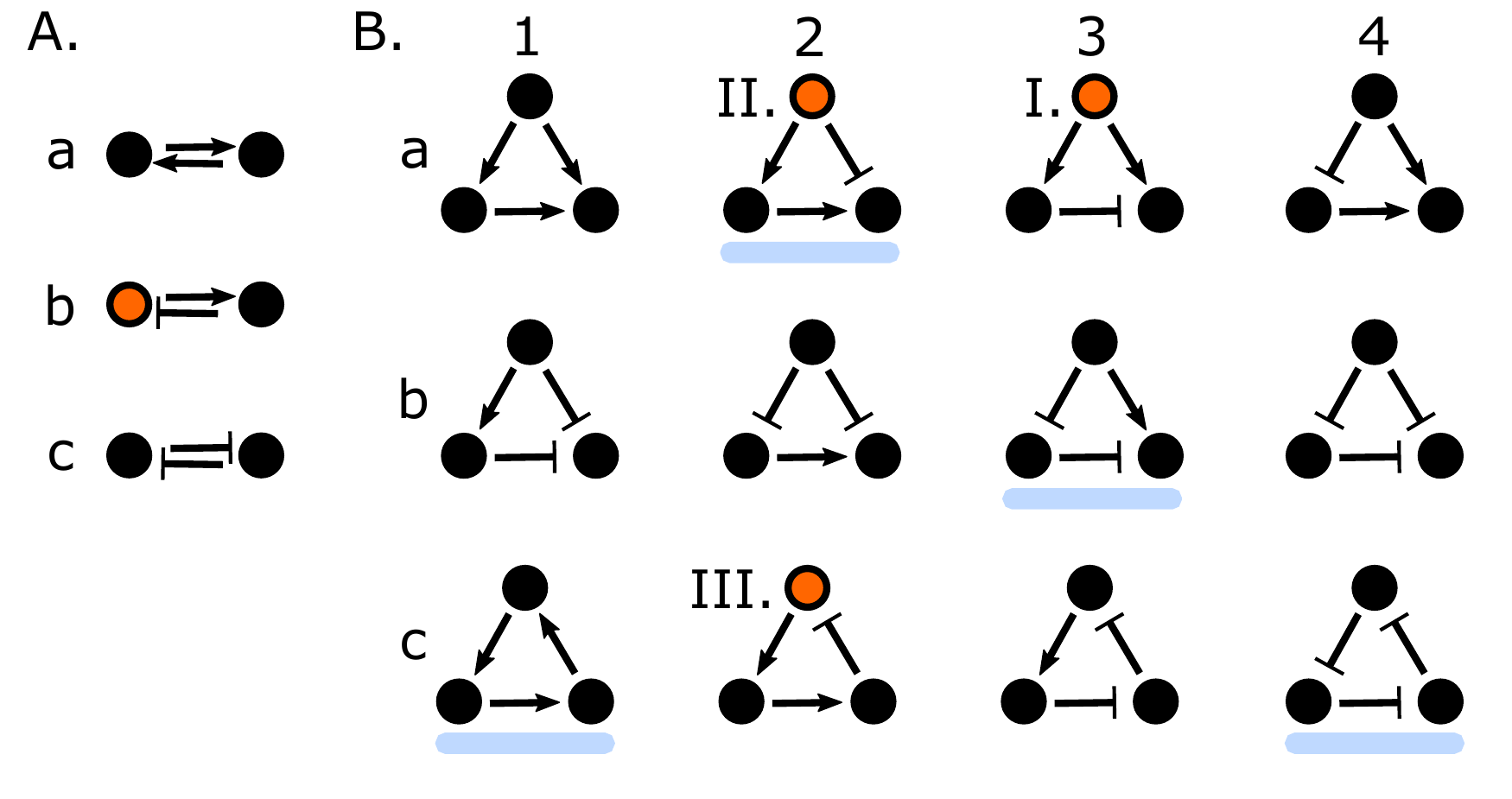}
    \caption{Ignoring symmetries, the possible edges on two nodes (A) and three nodes (B) with nodes implicated by an inconsistency highlighted in orange. Inconsistent triangles are different from those underlined in blue that disobey a hierarchical relationship.} 
    \label{fig:susp_motifs}
\end{figure}

In addition to being useful to an observer in a first pass, inconsistent motifs should also be rare in settings where nodes want to appear trustworthy, such as the online trading of Bitcoin, a decentralized cryptocurrency. The exchanges which oversee this trading often allow users to rate the trustworthiness of other users \cite{kumar2016edge,kumar2018rev2} due to the inherent risk of trading Bitcoin induced by the presence of fraud \cite{moore2013beware}.  Because users are particularly interested in appearing honest in order to gain counterparties to future trades, these exchanges are an ideal practical example of networks where nodes might want to appear consistent to an external observer. Two examples of such networks are the Bitcoin OTC and Bitcoin Alpha networks \cite{kumar2016edge,kumar2018rev2} which describe two different Bitcoin exchanges as signed and weighted networks. Table \ref{tab:Tri_OTC} shows a variety of statistics for the occurrence of the different triangles in Figure \ref{fig:susp_motifs} in each network.

Indeed, the suspicious motifs in the Bitcoin networks are underrepresented compared to an Erd\H{o}s-R\'{e}nyi null model, while the motifs one would expect to arise out of altering them, as described in the next section and displayed in Figure~\ref{fig:motif_changes}, are over-represented\footnote{As in many empirical graphs, all triangles are over-represented compared to an Erd\H{o}s-R\'{e}nyi null model, but this this null model controls for each motif's relative over-representation.}. Further, while consistency makes sense as a driving force in these Bitcoin networks, the hypothesized forces of structural balance and status would seem to be less important, and this is evidenced by the data. One notable difference between the observed triangle frequencies and the predictions of structural balance theory is that structural balance theory predicts that triangles with three accusations (4b \& 4c) should be rare, but these triangles are, relative to the null, highly over-represented in both Bitcoin OTC and Alpha. Similarly, status theory, which predicts that edges a and c and triangles 1c, 2a, 3b and 4c should be rare, is inconsistent with the prevalence of edges a and c as well as that of triangles 4c and 3b.  

\begin{table}[]
\centering
\begin{tabular}{cc|cccccc}
   & Edge $e_i$ & $ |e_i|$    & $\mathbb{E}|e_i|$ & $ \sum W(e_i)$ & $\mathbb{E} \sum W(e_i)$ & $\frac{1}{M}\frac{|T_i|}{\mathbb{E} |e_i|}$ & $ \frac{1}{M}\frac{\sum W(T_i) }{ \mathbb{E} \sum W(T_i)}$ \\
  \parbox[t]{2mm}{\multirow{3}{*}{\rotatebox[origin=c]{90}{ OTC}}}                              & a          & 13438        & 14.833            & 22019          & 17.950                   & 1.00                                        & 1.00                                                       \\
   & \textbf{b} & \textbf{358} & 3.300             & \textbf{743}   & 3.992                    & \textbf{0.12}                               & \textbf{0.15}                                              \\ 
   & {\ul c}    & {\ul 304}    & 0.184             & {\ul 2224}     & 0.222                    & {\ul 1.83}                                  & {\ul 8.16}                                             \\    [0.5ex] 
\hline\hline   
       \parbox[t]{2mm}{\multirow{3}{*}{\rotatebox[origin=c]{90}{ Alpha}}} & a          & 9678         & 17.929            & 16325          & 22.046                   & 1.00                                        & 1.00                                                   \rule{0pt}{3ex}    \\
       & \textbf{b} & \textbf{248} & 2.432             & \textbf{495}   & 2.989                    & \textbf{0.19}                               & \textbf{0.22}                                              \\
       & {\ul c}    & {\ul 136}    & 0.082             & {\ul 859}      & 0.101                    & {\ul 3.06}                                  & {\ul 11.44}                                               
\end{tabular}
\vspace{5mm}

\begin{tabular}{cc|cccccc}
                   & Triad $T_i$ & $ |T_i|$     & $\mathbb{E}|T_i|$ & $ \sum W(T_i)$ & $\mathbb{E} \sum W(T_i)$ & $\frac{1}{M} \frac{|T_i|}{\mathbb{E} |T_i|}$ & $\frac{1}{M} \frac{\sum W(T_i)}{ \mathbb{E} \sum W(T_i)}$ \\
\parbox[t]{2mm}{\multirow{12}{*}{\rotatebox[origin=c]{90}{Bitcoin OTC}}} & {\ul 1a}    & {\ul 103194}  & 161.539           & {\ul 133449}   & 171.820                  & {\ul 0.99}                                   & {\ul 0.88}                                                \\
& \textbf{2a} & \textbf{3752} & 17.970            & \textbf{5175}  & 20.999                   & \textbf{0.32}                                & \textbf{0.28}                                             \\
& \textbf{3a} & \textbf{3883} & 17.970            & \textbf{5316}  & 20.999                   & \textbf{0.33}                                & \textbf{0.29}                                             \\
& 4a          & 4670          & 17.970            & 6850           & 20.999                   & 0.40                                         & 0.37                                                      \\
& {\ul 1b}    & {\ul 4391}    & 1.999             & {\ul 11020}    & 3.173                    & {\ul 3.40}                                   & {\ul 3.92}                                                \\
& {\ul 2b}    & {\ul 3921}    & 1.999             & {\ul 12821}    & 3.173                    & {\ul 3.04}                                   & {\ul 4.56}                                                \\
& {\ul 3b}    & {\ul 1478}    & 1.999             & {\ul 5150}     & 3.173                    & {\ul 1.14}                                   & {\ul 1.83}                                                \\
& 4b          & 597           & 0.222             & 3189           & 1.007                    & 4.16                                         & 3.57                                                      \\
& 1c          & 32651         & 53.846            & 41304          & 57.273                   & 0.94                                         & 0.81                                                      \\
& \textbf{2c} & \textbf{4563} & 17.970            & \textbf{6542}  & 20.999                   & \textbf{0.39}                                & \textbf{0.35}                                             \\
& {\ul 3c}    & {\ul 1305}          & 1.999             & {\ul 3990}           & 3.173                    & {\ul 1.01}                                   & {\ul 1.42}                                                \\
& 4c          & 62            & 0.074             & 334            & 0.336                    & 1.30                                         & 1.12                                                       \\    [0.5ex] 
 \hline\hline  
\parbox[t]{2mm}{\multirow{12}{*}{\rotatebox[origin=c]{90}{Bitcoin Alpha}}} & {\ul 1a}    & {\ul 74632}   & 214.633           & {\ul 100005}   & 230.303                  & {\ul 0.77}                                   & {\ul 0.47}                                            \rule{0pt}{2.5ex}    \\
& \textbf{2a} & \textbf{2921} & 14.555            & \textbf{3857}  & 16.889                   & \textbf{0.44}                                & \textbf{0.25}                                             \\
& \textbf{3a} & \textbf{3035} & 14.555            & \textbf{4095}  & 16.889                   & \textbf{0.46}                                & \textbf{0.26}                                             \\
& 4a          & 3872          & 14.555            & 5534           & 16.889                   & 0.59                                         & 0.36                                                      \\
& {\ul 1b}    & {\ul 2492}          & 0.987             & {\ul 5730}     & 1.423                    & {\ul 5.58}                                   & {\ul 4.37}                                                \\
& {\ul 2b}    & {\ul 951}     & 0.987             & {\ul 4383}     & 1.423                    & {\ul 2.13}                                   & {\ul 3.34}                                                \\
& {\ul 3b}    & {\ul 550}     & 0.987             & {\ul 2016}     & 1.423                    & {\ul 1.23}                                   & {\ul 1.54}                                                \\
& 4b          & 300           & 0.067             & 1405           & 0.200                    & 9.91                                         & 7.63                                                      \\
& 1c          & 23717         & 71.544            & 30933          & 76.768                   & 0.73                                         & 0.44                                                      \\
& \textbf{2c} & \textbf{3806} & 14.555            & \textbf{5355}  & 16.889                   & \textbf{0.58}                                & \textbf{0.34}                                             \\
& {\ul 3c}    & {\ul 597}     & 0.987             & {\ul 2004}     & 1.423                    & {\ul 1.34}                                   & {\ul 1.53}                                                \\
& 4c          & 31            & 0.022             & 168            & 0.067                    & 3.07                                         & 2.74                                                     
\end{tabular}    \caption{The number of edges (top) and triangles (bottom) observed in the Bitcoin OTC and Bitcoin Alpha networks using motif names from Figure \ref{fig:susp_motifs}. $|e_i|$ and $|T_i|$ are counts of the motifs, and $W(e_i)$ and $W(T_i)$ incorporate the weights in the graphs where the weight of the motif is the minimum weight of the edges contained in the motif.
$\mathbb{E}|e_i|$, $\mathbb{E}|T_i|$, $\mathbb{E}W(e_i)$ and $\mathbb{E}W(T_i)$ are the expected number and weights of these motifs under an Erd\H{o}s-R\'{e}nyi null model. 
$\frac{|x|}{\mathbb{E} |x|}$ is the ratio of observed motifs to expected motifs and $\frac{1}{M} \frac{|x|}{\mathbb{E} |x|}$ is this quantity scaled about it's median $M$ while $\frac{1}{M} \frac{\sum W(x)}{ \mathbb{E} \sum W(x)}$ is similarly calculated for the weighted version. Motifs expected to be rare based on consistency are in bold while those which are produced by consistency dynamics are underlined.   }
    \label{tab:Tri_OTC}
\end{table}

\section{Implication Avoiding Dynamics}

In a context where nodes are aware that they are being scrutinized by an external observer using inconsistent motifs as a heuristic for identifying suspicious nodes, how might nodes attempt to appear credible? To address this question, we propose a set of node behaviors, collectively called implication avoiding dynamics (IAD), where nodes dynamically change their edges so as to be involved in no inconsistent motifs. 

\subsection{Local Implication Avoiding Dynamics}\label{sec:local_iad}

Consider first the situation where nodes have access to local knowledge about their own and their neighbors' edges. In this setting nodes will be concerned about their involvement in inconsistent edges and triangles and will attempt to alter their outgoing edges in order to avoid participation in these arrangements, as shown in Figure \ref{fig:motif_changes}. 
\begin{figure}
    \centering
    \includegraphics[width=\textwidth]{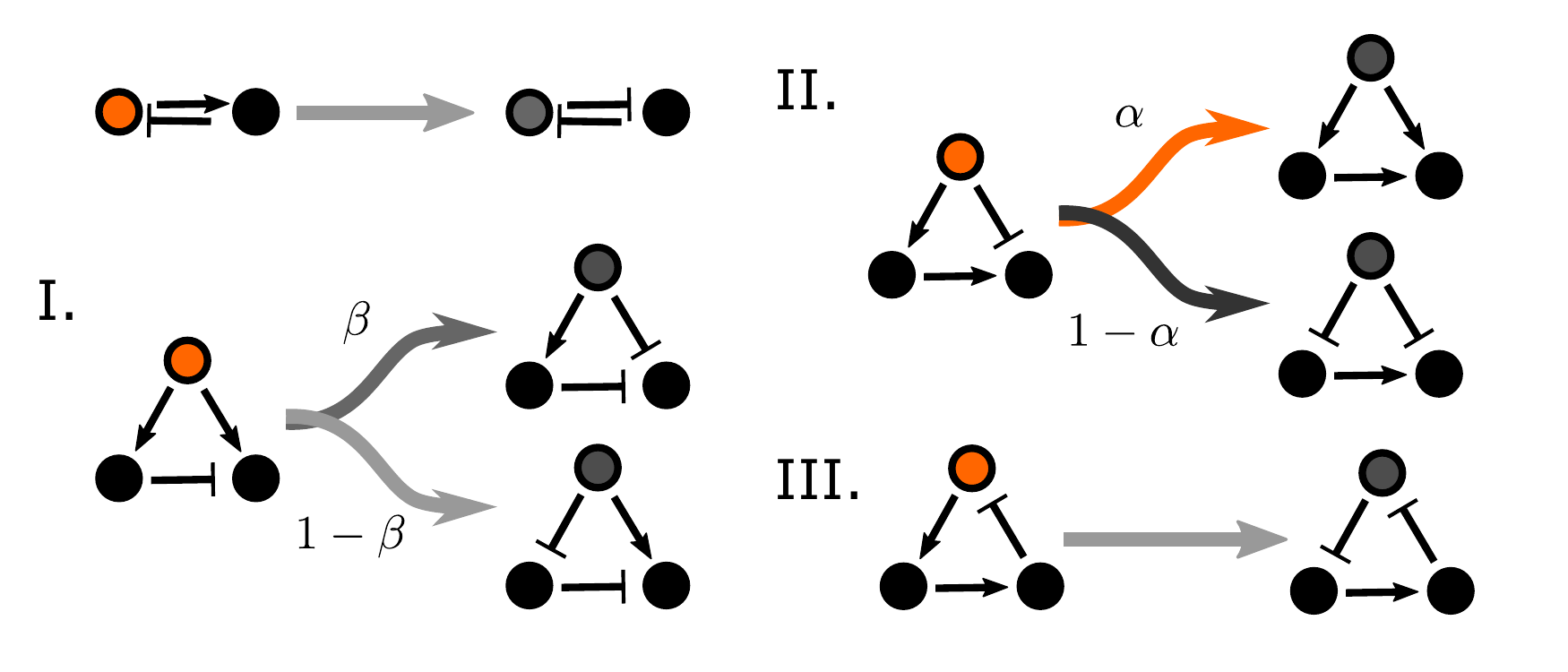}
    \caption{Local implication avoiding dynamics involve nodes implicated in their egonet changing their edges to become consistent. Changes that create accusations are labeled with grey arrows, while the one motif change that creates a new endorsement is orange. Parameters $\alpha$ and $\beta$ determine the likelihood of changes when two possibilities are available.}
    \label{fig:motif_changes}
\end{figure}

This suggests a discrete time dynamical system, where inconsistent edges, being very pronounced, are always instantly resolved, and inconsistent triangles are iteratively resolved. Thus, at each time step a node selected uniformly at random, randomly resolves one of the inconsistent triangles it is implicated by, according to the transitions in Figure \ref{fig:motif_changes}. Notice that for both Type I and Type II triangles the node has a choice of which edge to change, and the probability of these changes are determined by parameters $\alpha$, and $\beta$. In particular, $\alpha$ is the propensity of a player to change their accusation to an endorsement in a Type II motif while $\beta$ is the likelihood of siding with the accuser rather than the accused in a Type I motif. 

Notice that local implication avoiding dynamics can change endorsements and accusations into each other but never destroys an edge between nodes. Thus, local implication avoiding dynamics changes the nature of the relationships between nodes, but not the overall structure of relationships in the Network. In terms of equilibrium\footnote{Since we are considering IAD as dynamics we will discuss equilibrium states, but these will be the same as absorbing states of the Markov Chain associated with IAD. }, any network without any inconsistent triangles or edges is clearly invariant under these dynamics, so triangle free graphs with no inconsistent edges are automatically at equilibrium. There is also a clear description of the equilibrium state of these dynamics on a complete graph, where either $i\Dashv j$ or $i \to j$ for all node pairs $(i,j)$:
\begin{theorem}
A complete signed directed graph, is in equilibrium under local implication avoiding dynamics if and only if it can be separated into $k$ disjoint, self-consistent, insular communities. 
\label{thm:complete_net_equil_for_local_iad}
\end{theorem}
\begin{proof} If the graph is then made up of $k$ disjoint, self-consistent and insular communities, it necessarily does not have any inconsistent triangles. Namely, the only possible triangles in such a graph are comprised of 3 nodes belonging to same community, two nodes belonging to the same community with one belonging to a different community, or three nodes belonging to different communities, none of which are suspicious motifs. Next, consider a graph in equilibrium under implication avoiding dynamics. Clearly all edges are reciprocated in kind, so we may consider our directed graph as a simple graph. If a node endorses two other nodes, stability with regard to triangle I dictates that those nodes can't accuse each other, implying that in a complete graph they endorse each other. Iterating this argument implies that the component connected by endorsements to any node $u$, which in a simple graph is insular by construction, is also self-consistent, and the set of all such disjoint components comprises our desired $k$ communities. \end{proof}

Notice, self-consistent, insular communities were precisely the communities of most interest in Theorem \ref{thm:tree}, and they now appear as a natural consequence of local implication avoiding dynamics in complete networks. Similarly, since a complete graph is at equilibrium under structural balance if and only if it can be split into two disjoint mutually antagonistic communities \cite{cartwright1956structural,antal2005dynamics} then the complete graph equilibrium states of structural balance are a proper subset of those of IAD (when $n\ge 3$).

The equilibrium conditions for sparse graphs are harder to describe explicitly for local implication avoiding dynamics, other than that they obviously lack inconsistent triangles and edges. For example, the corresponding version of Theorem~\ref{thm:complete_net_equil_for_local_iad} for non-complete graphs is only the following:
\begin{corollary}
A graph composed of $k$ disjoint, self-consistent, insular communities is in equilibrium.
\end{corollary} \label{cor:convergence_of_liad}
This implies that a graph split into two self-consistent communities that mutually accuse each other is in equilibrium.

While equilibrium structures for non-complete networks can be complicated, in most settings the long term behavior of local implication avoiding dynamics will reach equilibrium.

\begin{theorem}
 If $\alpha\in[0,1)$, then a graph governed by local implication avoiding dynamics eventually converges to an equilibrium state.
\end{theorem} \label{thm:convergence_of_liad}
\begin{proof}First, consider that if $\alpha = 0$, the dynamics are only capable of changing endorsements into accusations. Consequently, the total number of accusations in the graph is non-decreasing and bounded above by $n^2$, implying that from any non-equilibrium state, there exists at least one sequence of fewer than $n^2$ steps that takes the system to equilibrium and this sequence occurs with probability greater than $p^{n^2}\frac{1}{6}^{n^2}$, where $p$ is the per step probability of sampling any given triangle. If $0<\alpha<1$, the same sequence of steps remains possible, and will occur with probability greater than $(1-\alpha)^{n^2}p^{n^2}\frac{1}{6}^{n^2}$. Since this probability is bounded below, the system eventually reaches equilibrium. \end{proof}

In contrast, when $\alpha=1$, infinite cycles of alternating edge changes are possible, for example the one in Figure \ref{fig:equil_cycles}, where three states are repeatedly cycled between, and the two absorbing states are never visited. 

\begin{figure}
    \centering
    \includegraphics[width=\textwidth]{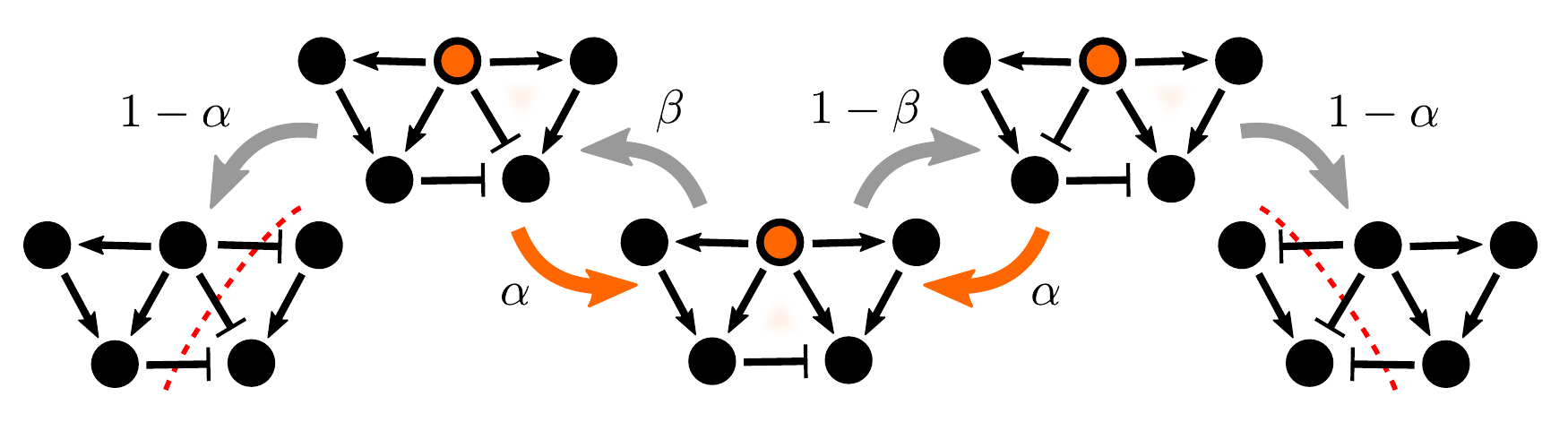}  
    \caption{A situation in IAD in which the system does not achieve equilibrium when $\alpha=1$, as the system traverses the middle three states but cannot transition to the outer, stable states.}
    \label{fig:equil_cycles}
\end{figure}

\subsection{Simulations} \label{sec:simulations}
Given the large number of possible equilibrium structures in local IAD, enumerating equilibrium states does not by itself shed light on the outcome of the dynamics. Numerical simulation suggests that local IAD on sparse networks often results in equilibrium states whose structures are difficult to summarize, but it leaves a distinct and pronounced effect on the relative proportion of endorsements to accusations. Indeed, of the inconsistent motifs shown in Figure \ref{fig:susp_motifs} only one of the possible resolutions involves making a new endorsement, while the rest all make accusations. The result of this is that accusations tend to increase as a network evolves under implication avoiding dynamics.

When the network is composed of a set of honest nodes $H$ and a cheaters $C$ then, since $H$ is perfectly consistent, IAD only affects nodes in $C$ and typically leads to an increasing number of accusations internal to $C$, effectively fracturing the endorsement structure of $C$.  

More generally, IAD is often a dynamics of community fracture. To illustrate this, consider an Erd\H{o}s-R\'enyi endorsement network of $30$ nodes and $p=27\%$. Such an endorsement graph represents a relatively close-knit community, yet if a single accusation is introduced IAD leads the community to become a web of accusations, as in Figure \ref{fig:fracture}. This skew towards accusation occurs even when the dynamics are heavily skewed towards resolving triangle II by creating an endorsement ($\alpha=.9$ in Figure \ref{fig:fracture}).

\begin{figure}
    \centering
    \includegraphics[width=1\textwidth]{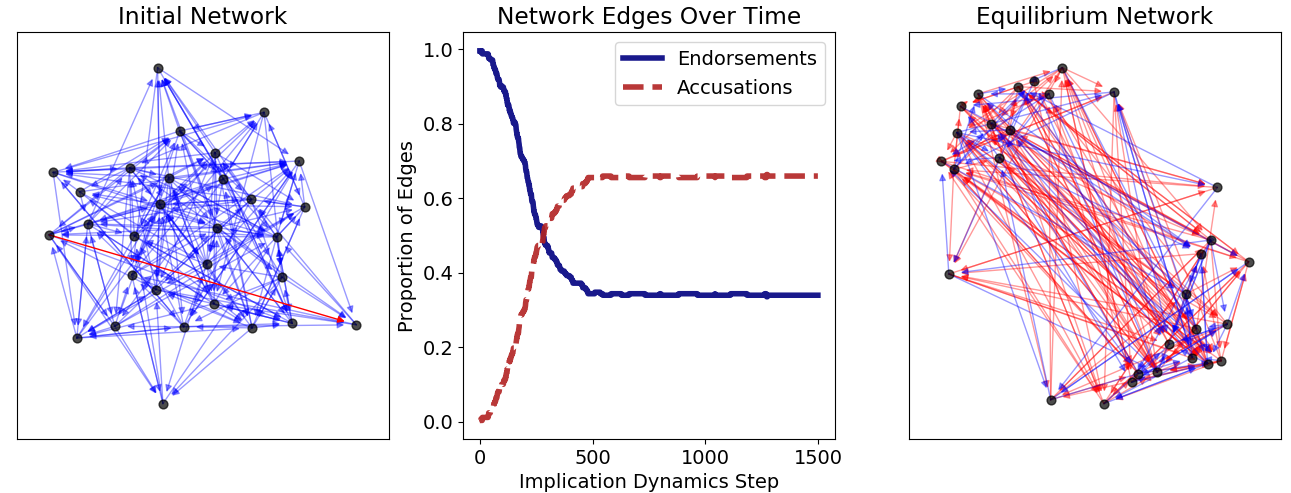}
    \caption{An illustration of the dynamics of fracture caused by local implication avoiding dynamics. A small, mostly friendly community with a single negative edge (left) swiftly devolves into a fractured, antagonistic network (right), with the number of accusations eventually exceeding the number of endorsements (center). Parameters where, $\alpha=.9$, $\beta = .5$.(Color available online.)}
    \label{fig:fracture}
\end{figure}

Interestingly, whether or not such a fracturing occurs, as well as the extent to which the network splits, depends on the density of edges and triangles in the network.  For instance, in an Erd\H{o}s-R\'enyi endorsement network the final proportion of accusations following IAD grows with the number of total edges. In particular, there is a second-order phase transition after which the number of accusations will exceed the number of endorsements in the equilibrium state for local IAD, as shown in Figure \ref{fig:phase_trans}. Such a result fits with common intuition that highly connected networks have more acrimonious fracturing than loosely connected ones.

\begin{figure}
    \centering
    \includegraphics[width=0.75\textwidth]{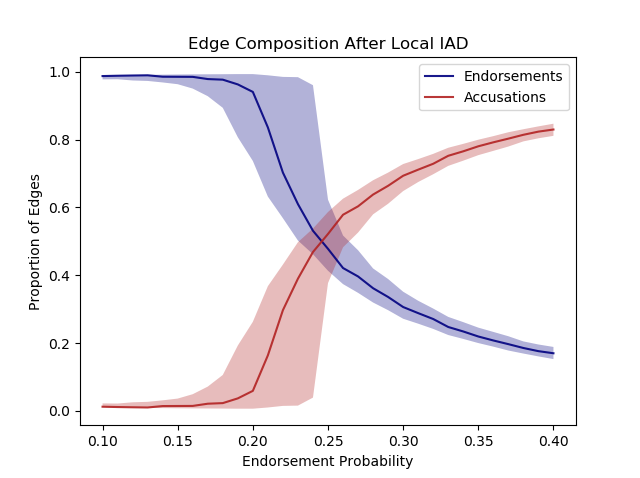}
    \caption{Relative proportion of signed edges in an Erd\H{o}s-R\'enyi endorsement network with a single accusation after evolution to equilibrium under local IAD for varying endorsement probabilities. The median and interquartile range of 1000 networks each run for 1500 iterations is plotted.}
    \label{fig:phase_trans}
\end{figure}

\subsection{Strong Implication Avoiding Dynamics}\label{sec:strong_iad}

In situations where nodes have complete information on the state of the graph, nodes may be seen as responsible for suspicious motifs involving longer endorsements paths. As a result, a node may now also alter its edges in order to resolve suspicious motifs at a variety of depths, giving rise to `strong' IAD, as shown in Figure \ref{fig:strong_changes}. A natural assumption would be that shorter range motifs resolve at a much faster rate, or entirely before, longer range motifs. This suggests that strong IAD can be thought of as a process on top of local IAD.

\begin{figure}
    \centering
    \includegraphics[width=\textwidth]{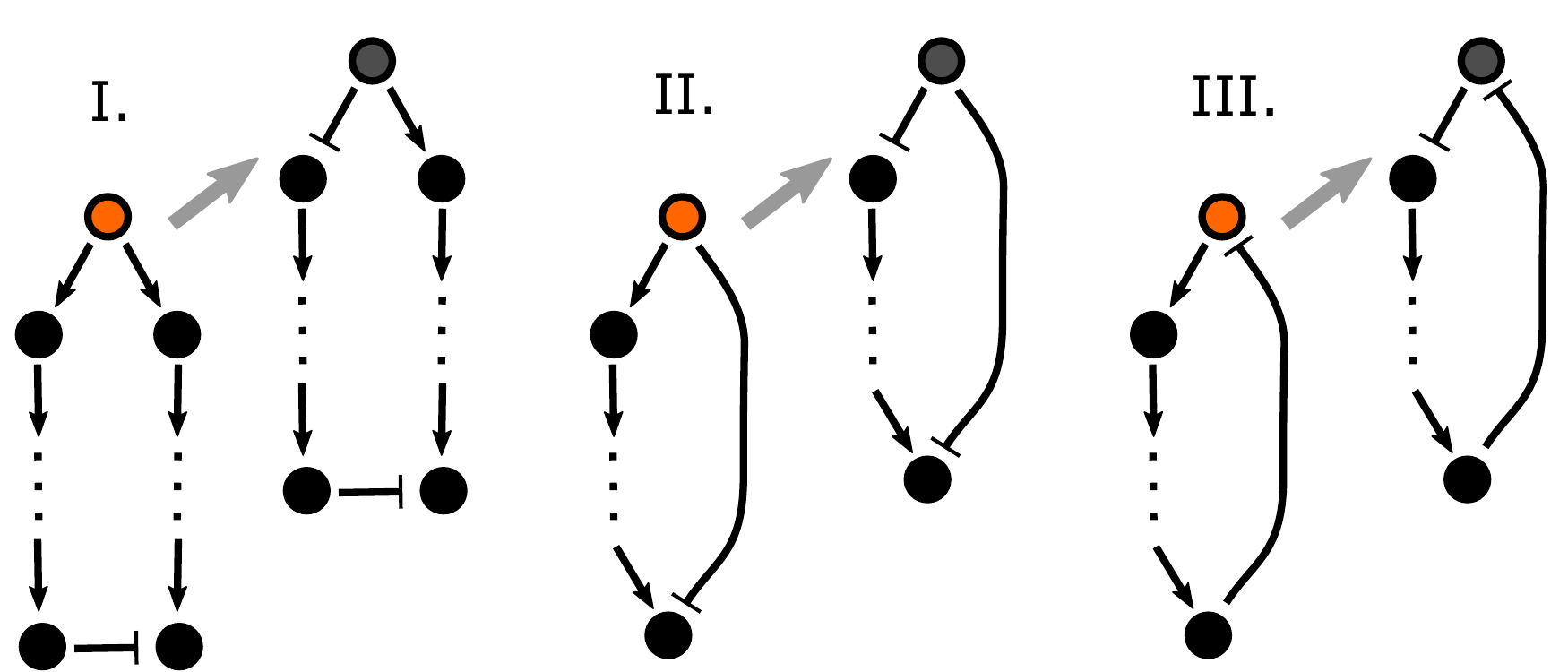}  
    \caption{Strong implication avoiding dynamics rewire edges so that nodes appear consistent at all depths.}
    \label{fig:strong_changes}
\end{figure}

One important observation about these strong implication avoiding dynamics is that they are most useful for sparse networks. In the case of complete networks the strong dynamics yield the same equilibrium state as the local dynamics, since a lack of any local inconsistencies implies a lack of any non-local inconsistencies, and vice versa. 
Indeed, while local IAD did not have a useful and succinct description of it's equilibrium state for sparse networks, the equilibrium state of strong implication avoiding dynamics does.  Namely, a network at equilibrium under strong IAD can be decomposed into several endorsement insular structures which involve no accusations between any nodes in $u, v \in \sigma(w)$ for all $w$, and no accusations between any $u$ and $v$ with $u\in \sigma(v)$, as in Figure \ref{fig:equil_structure}. 

\begin{figure}
    \centering
    \includegraphics[width=\textwidth]{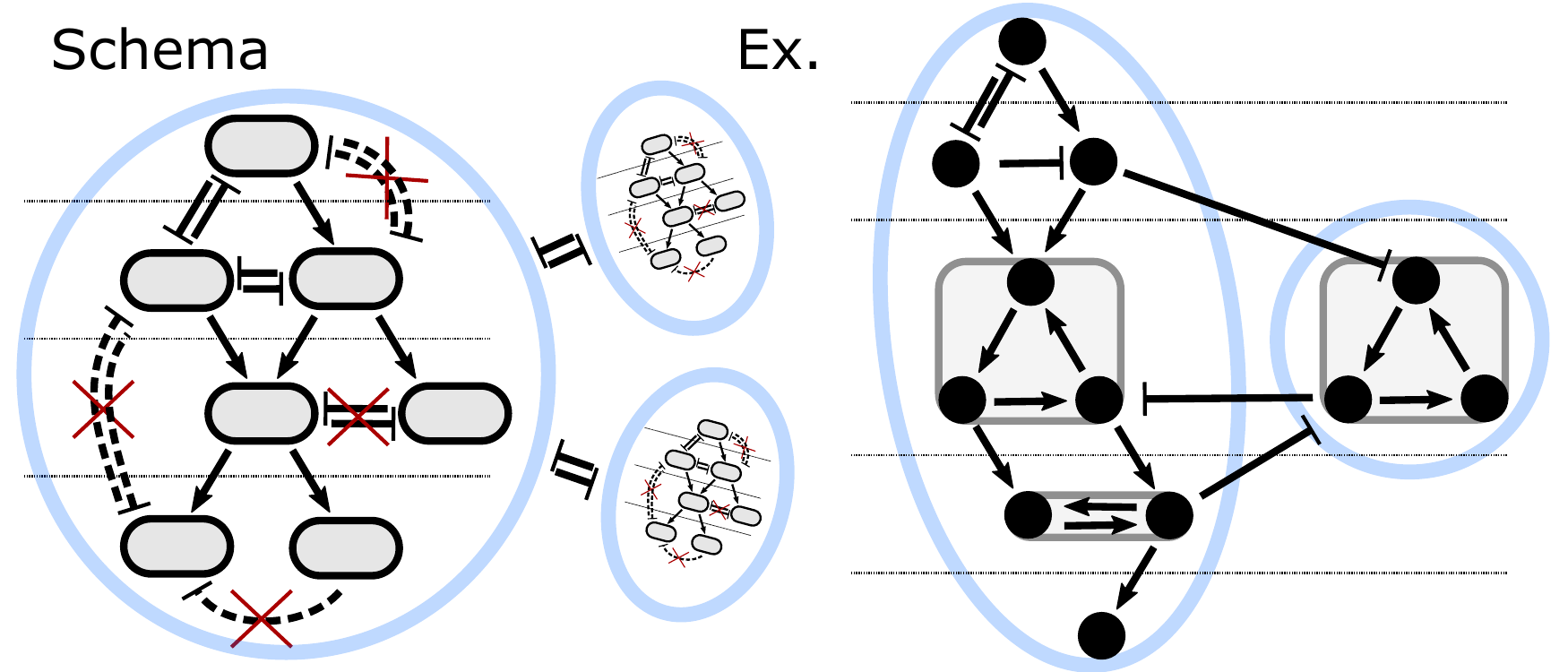} 
    \caption{A network at equilibrium  under strong implication avoiding dynamics can be separated into multiple insular trees, possibly with accusations between trees. While these trees may not be self-consistent, the allowable positions of internal accusations can be viewed in the topological sorting of the tree's condensation under endorsements (left), where each condensed meta-node is a strongly connected component of $G^\to$. For example, the graph on the right fits the schema on the left. }
    \label{fig:equil_structure}
\end{figure}

While these results are not as strong as those in Theorem \ref{thm:complete_net_equil_for_local_iad}, the equilibrium state of strong implication avoiding dynamics strikes towards the key elements of the hypotheses of Theorem \ref{thm:tree} and in a way that remains functional even on sparse, relatively error-prone networks.

\section{In-Groups and Out-Groups}

While this paper has thus far discussed `honest' nodes and `cheaters' the approaches of this paper apply equally well to scenarios where nodes wish to be seen in an `in-group' and avoid being detected as a member of an `out-group'.  In particular, there are two recursive definitions of `in-group' membership that correspond to IAD on a complete graph and to Strong IAD on sparse networks respectively.
\begin{definition}[Complete Information In-Group]
A member of a complete information in-group:
\begin{enumerate}
    \item can identify all other members of the complete in-group, 
    \item will truthfully report their identifications. 
\end{enumerate}
\end{definition}
This definition of a complete in-group implies that the in-group comprise a self-consistent, insular clique, and would only be suitable when it's common knowledge that the in-group is small enough that all its members know each other. When an in-group is large enough, or opaque enough, that in-group members will only know a subset of the in-group, the the following definition is more appropriate.
\begin{definition}[Partial Information In-Group]
A member of a partial information in-group:
\begin{enumerate}
    \item will not report an out-group member as in-group, 
    \item will not report an in-group member as out-group. 
\end{enumerate}
\end{definition}
A partial information in-group looks like the equilibrium structure of strong IAD, as in Figure~\ref{fig:equil_structure}. Notice that both these definitions work when there are multiple distinct notions of in-group, leading to multiple different communities.  

While these notions of in-group and out-group have somewhat clean equilibrium structures, in practice people almost surely have some amount of disagreement on the status and definition of in-group and out-group. Indeed, when there is disagreement between agents as to whom is in the in-group, both local and strong IAD predict an increase in the aggregate number of out-group accusations, as demonstrated in Section~\ref{sec:simulations}. 

An interesting prediction of these dynamics is that when the in-group is in equilibrium, IAD can lead to the out-group to become ever more fractured in a vain attempt to appear as if it were a  consistent in-group node.

\section{Conclusions and Future Work}

We have considered a network construction game where nodes create signed directed edges in order to convince an external observer that they are a member of the honest nodes. A preliminary examination of this setting revealed that the adversarial nodes can coordinate to confuse the observer if they number at least as many as the honest nodes. When honest nodes outnumber the adversarial nodes they can create identifiable structures through either endorsements or accusations. While these theorems establish the theoretical possibility of the observer's task, they are fragile if the honest nodes erroneously endorse a cheater or accuse another honest node and in the worst are computationally impractical. Instead, there exists an interesting set of heuristics, which look for inconsistent motifs, and use these motifs to implicate individual nodes, suggesting an implication avoiding dynamics where nodes adjust their outgoing edges to resolve any inconsistent motifs. Interestingly, these implication avoiding dynamics 
 typically increase the number of accusations they make, providing a natural dynamics of network fracture. This provides an interesting contrast with Jeremy Bentham's hypothesized Panopticon, where an external observer is hypothesized to correct undesired behavior. Instead, this setting shows a situation where external observation leads individual behavior to be increasingly accusatory, and previously healthy trusting networks can fracture into suspicion and paranoia. Perhaps, intuitively, when this happens, the effects are strongest in closer-knit groups. 

While IAD may be able to explain some network fractures, it neither can nor attempts to explain the original creation of a network. For this reason, the application of this work to situations where people sort themselves into various in-groups and out-groups may provide useful insights in empirical work, but would likely require careful thinking about the meta-dynamics of what triggers the beginning and end of IAD. For example, the announcement of an internal investigation, the revelation of a scandal, the arrival of a new community members or the election of a leader could each lead to situations where individuals are concerned about the consistency of their outgoing edges.     

One possible shortcoming of IAD is that nodes, in pursuit of appearing completely consistent, will make edge changes that temporarily increase the number of the inconsistent triangles they appear in. In contrast, future studies may wish to consider a modification of IAD in which nodes greedily seek to increase some measure of consistency. Similarly, while the current dynamics necessarily increase the number of accusations, one can imagine a different dynamics where inconsistent triangles are eliminated so that the number of accusations and endorsements is kept fixed. For example, rather than changing an endorsement to an accusation, an implicated node could rewire that endorsement to a random node in the graph.  

Another interesting approach would be consider if the observer has an algorithm which supports some notion of pro-social behavior. We have shown one instance where the algorithmic choice of the observer can lead to an explosion of accusations, but it is conceivable that if the observer picked a different metric, they could emphasize a different set of behaviors. From this viewpoint, the observer would be picking an algorithm in a multi-objective setting, where the reliability of the algorithm in detecting an adversarial set of nodes must be balanced with the effects that this choice would have node behavior. Namely, IAD represents one of many possible choices the observer has, and it is a choice that heavily affects node behavior. Alternatively, and at other extreme, if the observer where to announce that they will randomly classified nodes, irrespective of the network, then node behavior would be completely unaffected by the observer. Whether there are interesting and useful algorithms in between, and whether this trade-off is best conceptualized as a linear trade-off or as inhabiting a multi-objective space would provide insight into the nature of this problem.

In each of these cases, careful thinking about the interplay between observer and actor could lead to insightful theories about network structure.

\section*{Acknowledgments}
This work was supported by a scholarship from the Barrett, The Honors College at Arizona State University and by a New College Undergraduate Inquiry and Research Experiences (NCUIRE) Fellowship from the New College of Interdisciplinary Arts and Sciences at Arizona State University.  

\bibliographystyle{plain}
\bibliography{references}

\begin{thebibliography}{10}

\bibitem{10.1371/journal.pone.0038135}
Claudio Altafini.
\newblock Dynamics of opinion forming in structurally balanced social networks.
\newblock {\em PLOS ONE}, 7(6):1--9, 06 2012.

\bibitem{antal2005dynamics}
Tibor Antal, Pavel~L Krapivsky, and Sidney Redner.
\newblock Dynamics of social balance on networks.
\newblock {\em Physical Review E}, 72(3):036121, 2005.

\bibitem{cartwright1956structural}
Dorwin Cartwright and Frank Harary.
\newblock Structural balance: a generalization of heider's theory.
\newblock {\em Psychological review}, 63(5):277, 1956.

\bibitem{cheng2005sybilproof}
Alice Cheng and Eric Friedman.
\newblock Sybilproof reputation mechanisms.
\newblock In {\em Proceedings of the 2005 ACM SIGCOMM workshop on Economics of
  peer-to-peer systems}, pages 128--132. ACM, 2005.

\bibitem{heider1946attitudes}
Fritz Heider.
\newblock Attitudes and cognitive organization.
\newblock {\em The Journal of psychology}, 21(1):107--112, 1946.

\bibitem{isakov2019structure}
Alexander Isakov, James~H Fowler, Edoardo~M Airoldi, and Nicholas~A Christakis.
\newblock The structure of negative social ties in rural village networks.
\newblock {\em Sociological Science}, 6:197--218, 2019.

\bibitem{kumar2018rev2}
Srijan Kumar, Bryan Hooi, Disha Makhija, Mohit Kumar, Christos Faloutsos, and
  VS~Subrahmanian.
\newblock Rev2: Fraudulent user prediction in rating platforms.
\newblock In {\em Proceedings of the Eleventh ACM International Conference on
  Web Search and Data Mining}, pages 333--341. ACM, 2018.

\bibitem{kumar2016edge}
Srijan Kumar, Francesca Spezzano, VS~Subrahmanian, and Christos Faloutsos.
\newblock Edge weight prediction in weighted signed networks.
\newblock In {\em 2016 IEEE 16th International Conference on Data Mining
  (ICDM)}, pages 221--230. IEEE, 2016.

\bibitem{leskovec2010predicting}
Jure Leskovec, Daniel Huttenlocher, and Jon Kleinberg.
\newblock Predicting positive and negative links in online social networks.
\newblock In {\em Proceedings of the 19th international conference on World
  wide web}, pages 641--650. ACM, 2010.

\bibitem{leskovec2010signed}
Jure Leskovec, Daniel Huttenlocher, and Jon Kleinberg.
\newblock Signed networks in social media.
\newblock In {\em Proceedings of the SIGCHI conference on human factors in
  computing systems}, pages 1361--1370. ACM, 2010.

\bibitem{liu2016personalized}
Brandon~K Liu, David~C Parkes, and Sven Seuken.
\newblock Personalized hitting time for informative trust mechanisms despite
  sybils.
\newblock In {\em Proceedings of the 2016 International Conference on
  Autonomous Agents \& Multiagent Systems}, pages 1124--1132. International
  Foundation for Autonomous Agents and Multiagent Systems, 2016.

\bibitem{marvel2011continuous}
Seth~A Marvel, Jon Kleinberg, Robert~D Kleinberg, and Steven~H Strogatz.
\newblock Continuous-time model of structural balance.
\newblock {\em Proceedings of the National Academy of Sciences},
  108(5):1771--1776, 2011.

\bibitem{marvel2009energy}
Seth~A Marvel, Steven~H Strogatz, and Jon~M Kleinberg.
\newblock Energy landscape of social balance.
\newblock {\em Physical review letters}, 103(19):198701, 2009.

\bibitem{mishra2011finding}
Abhinav Mishra and Arnab Bhattacharya.
\newblock Finding the bias and prestige of nodes in networks based on trust
  scores.
\newblock In {\em Proceedings of the 20th international conference on World
  wide web}, pages 567--576. ACM, 2011.

\bibitem{moore2013beware}
Tyler Moore and Nicolas Christin.
\newblock Beware the middleman: Empirical analysis of bitcoin-exchange risk.
\newblock In {\em International Conference on Financial Cryptography and Data
  Security}, pages 25--33. Springer, 2013.

\bibitem{rawlings2017structural}
Craig~M Rawlings and Noah~E Friedkin.
\newblock The structural balance theory of sentiment networks: Elaboration and
  test.
\newblock {\em American Journal of Sociology}, 123(2):510--548, 2017.

\bibitem{resnick2009sybilproof}
Paul Resnick and Rahul Sami.
\newblock Sybilproof transitive trust protocols.
\newblock In {\em Proceedings of the 10th ACM conference on Electronic
  commerce}, pages 345--354. ACM, 2009.

\bibitem{tang2010hybrid}
Jie Tang, Sven Seuken, and David~C Parkes.
\newblock Hybrid transitive trust mechanisms.
\newblock In {\em Proceedings of the 9th International Conference on Autonomous
  Agents and Multiagent Systems: volume 1-Volume 1}, pages 233--240.
  International Foundation for Autonomous Agents and Multiagent Systems, 2010.

\bibitem{van2011micro}
Arnout Van~de Rijt.
\newblock The micro-macro link for the theory of structural balance.
\newblock {\em The Journal of Mathematical Sociology}, 35(1-3):94--113, 2011.

\bibitem{wasserman1994social}
Stanley Wasserman and Katherine Faust.
\newblock {\em Social network analysis: Methods and applications}, volume~8.
\newblock Cambridge university press, 1994.

\end{thebibliography}

\end{document}